\documentclass[journal]{IEEEtran}
\usepackage{dannychen}
\usepackage{amsmath,amssymb,amsfonts,mathtools}
\usepackage{algorithm}
\usepackage{algpseudocode}

\usepackage{graphicx}
\usepackage{textcomp}

\newcommand{\dt}{\Delta t}
\newcommand{\diag}{\textnormal{diag}}

\newcommand{\Na}{\text{Na}^+}
\newcommand{\K}{\text{K}^+}

\newcommand{\blue}[1]{{#1}}
\renewcommand{\red}{{}}

\title{Explicitly Solvable Continuous-time Inference for Partially Observed Markov Processes}

\author{Daniel Chen, Alexander G.~Strang, Andrew W.~Eckford \IEEEmembership{Senior Member, IEEE}, Peter J.~Thomas 
\thanks{This work was supported in part by National Institutes of Health BRAIN Initiative grant R01 NS118606 and a National Science Foundation grant DMS-2052109 to PJT, as well as research support from the Oberlin College Libraries, and an NSERC Discovery grant to AWE.} 
\thanks{D. Chen and P.~J. Thomas are  with the Department of Mathematics, Applied Mathematics, and Statistics; Department of Electrical, Control and Systems Engineering; Department of Computer and Data Science; Department of Biology, Case Western Reserve University, Cleveland, OH 44106 USA (e-mail: txc461/pjthomas@case.edu).}
\thanks{A.~G.~Strang is with Department of Statistics, University of Chicago, Chicago, IL 60637 USA (email: alexstrang@uchicago.edu).} 
\thanks{A.~W.~Eckford is with Department of Electrical Engineering and Computer Science, York University, Toronto, ON M3J 1P3 Canada (e-mail: aeckford@yorku.ca).} }

\begin{document}
\maketitle

\begin{abstract}
Many natural and engineered systems can be modeled as discrete state Markov processes.  Often, only a subset of states are directly observable.
Inferring the conditional probability that a system occupies a particular hidden state, given the partial observation, is a problem with broad application.
In this paper, we introduce a continuous-time formulation of the sum-product algorithm, which is a well-known discrete-time method for finding the hidden states' conditional probabilities, given a set of finite, discrete-time observations. 
From our new formulation, we can explicitly solve for the conditional probability of occupying any state, given the transition rates and observations within a finite time window.
We apply our algorithm to a realistic model of the cystic fibrosis transmembrane conductance regulator (CFTR) protein for exact inference of the conditional occupancy probability, given a finite time series of partial observations.
\end{abstract}

\section{Introduction}

    

Markov processes---dynamic processes whose future behavior depends only on their present state---approximate a wide variety of natural and engineered systems.
Despite rapid advances in high-throughput data acquisition and data processing, many systems of interest contain important degrees of freedom that cannot be directly observed.
Inferring the conditional probability that such a \textit{partially observed Markov process} occupies specific hidden states, given the available observations, is a ubiquitous problem in science and engineering.
Examples appear in robotics \cite{grisettiyz2005improving}, ecology \cite{baum1967inequality},  neuroscience \cite{anderson2015stochastic}, and algorithmic text analysis \cite{blei2003latent}.

We are motivated by biological examples in the present paper.  
Ion channels in excitable membranes, such as the sodium ($\Na$) and potassium ($\K$) channels described in Hodgkin and Huxley's quantitative model for action potential generation  
\cite{hodgkin1952quantitative,skaugen1979firing,hilleionic}, provide an early example.  
Discrete state Markov models based on Hodgkin and Huxley's $\K$ channel contain five states, only one of which conducts an ionic current; the other states are ``silent" and cannot be distinguished by direct electrophysiological observation. 
Similarly, the $\Na$ channel has eight states: seven with zero conductance and one with nonzero conductance.
Colquhoun and Hawkes introduced maximum likelihood methods for inferring the rate constants of a partially observed Markov process representing the nicotinic Acetylcholine receptor \cite{colquhoun1977relaxation,colquhoun1981stochastic,colquhoun1982stochastic,qin1997maximum}, but did not address the question of inferring microscopic state occupancy from observable conductance time series.  
More recently, research into the molecular biology of cystic fibrosis (CF) has focused on the CF transmembrane conductance regulator (CFTR), which can be modeled as a 7-state system with two conducting states and five nonconducting states \cite{fuller2005block} (detailed below in Section \ref{sec:apps}). Beyond these biological examples, problems of inferring or estimating hidden states from incomplete observations are widely studied in the signal processing literature \cite{loeliger2007factor}.

The literature contains several approaches to approximating the behavior of hidden states of partially observed Markov processes.
Sampling provides one common technique for approximate inference
\cite{koller_friedman_2012}. 
As an example, recent work by Fang \textit{et al.}~demonstrated an efficient algorithm for simulating stochastic reaction networks with multiple separated time scales using particle filters \cite{fang2022stochastic}. 
In general, Markov Chain Monte Carlo is a widely employed sampling technique used to infer hidden states
\cite{cappe2003reversible, djuric1999estimation}\blue{, that has also been applied to ion channels \cite{rosales2001bayesian}.}
For partially observed Bayesian networks operating in discrete time, message passing algorithms on factor graphs provide an efficient and exact inference method  \cite{kschischang2001factor}. 
\red{The factor-graph formalism is highly flexible.} 
Algorithms based on the message passing concept have been extended to applications \red{in localization \cite{jin2020bayesian}, compressed sensing \cite{som2012compressive}, and decision fusion \cite{abrardo2017message}}.
\blue{Factor graphs are not limited to models with a discrete number of variables. 
For example, Gaussian message passing in linear Gaussian models (\textit{e.g.}~Kalman filtering and smoothing) has been developed for continuous-time models with discrete-time observations \cite{bolliger2012digital, bolliger2013lmmse, bruderer2014estimation}. 
However, the state reconstruction problem for continuous-time finite-state hidden Markov models has not been addressed in the literature, to be best of our knowledge.}

\blue{In this work, we extend the message-passing algorithm in order to analytically interpolate state-occupancy probabilities of a continuous-time system, given a discretely sampled time series.}
That is, we show how to infer the time-dependent conditional probabilities of latent states for continuous-time discrete-state \blue{homogeneous} Markov processes given a set of partial observations over a finite time window. 
We derive an equivalent formulation of the sum-product algorithm in continuous time that allows one to find an explicit analytic solution for the state occupancy probability. \red{Having explicit solutions lowers the computational cost and, unlike sampling-based approximate methods, does not sacrifice accuracy.}
\red{Furthermore, the continuous-time formalism leads to elegant simplification of the analytic solutions.}
For certain systems---like the three-state systems shown in Figure \ref{fig:three-state}---the conditional probability obeys a second-order inhomogeneous linear ordinary  differential equation. Finally, we demonstrate the practical functionality of the algorithm with the 7-state model for CFTR using simulated data. 

The paper is organized as follows. Section \ref{sec:theory} reviews the message-passing algorithm for Markov processes with binary observations. 
Section \ref{sec:main} displays the main result of the paper: a continuous-time formulation of the sum-product algorithm. 
\blue{We present the derivation of the conditional probability using the message-passing approach in continuous time, and the corresponding sum-product algorithm.}
Section \ref{sec:analytic} discusses the implications of the continuous-time formulation further. 
\red{We give examples of small systems to display how analytic solutions may be found. 
In Section \ref{sec:apps} we demonstrate the value of the algorithm for a larger, realistic system.

\section{Theory of the Sum-Product Algorithm}
\label{sec:theory}
The sum-product algorithm can be used generally for inference on probabilistic models that can be written as factor graphs \cite{kschischang2001factor}. There are many variants of the sum-product algorithm, each suitable for accomplishing a different task. For our purposes, we will focus on the forward/backward algorithm for inference on hidden Markov models. 


We consider a continuous-time, discrete-state homogeneous Markov process on a finite state space $\Omega$.
Given a discrete, uniformly spaced sampling interval, the continuous-time process induces a discrete-time Markov process specified by some column-stochastic transition matrix $P$ \blue{that is invariant in time}.
Let $\mathcal{S} \subset \Omega$ be a subset of states, and let $S_t \in \Omega$ be the state of the system at time $t$. 
We assume an observer can only see whether $S_t$ is in $\mathcal{S}$ or not.
Accordingly, let $\mathcal Y_t = m(S_t)$ represent the observable where $m(s): \Omega \rightarrow \{0,1\}$ is the indicator function for the set $\mathcal{S}$.
The goal of the algorithm is to infer the conditional probability of being at a particular state, $i \in \Omega$, given the binary observation. 

The algorithm involves three \red{vector-valued} quantities: a forward message $\bm \alpha_t$, a backward message $\bm \beta_t$, and the \red{observation message} $\bm \chi_t$. 
One can interpret the forward message as the probability of arriving at a certain state from time $0$ to time $t$ and the backward message as \blue{the likelihood} of occupying a certain state at time $t$ \red{conditioned on ending up in a given target state, or a given target set of states, at the end} of the measurement $T$. 
On the other hand, the \red{observation message} $\bm \chi_t$ is an indicator function of the possible states given the observation. 
For instance, if at time $t$, the observation of the total system were 
$\mathcal{Y}_t=1$, then $\bm \chi_t$ would be
a vector with $1$'s on the states in \blue{the observation set} $\mathcal S$ and $0$ otherwise. 
\red{Using superscripts to denote vector indices, \textit{i.e.} $\bm v \ind i$ denotes the $i$-th element of vector $\bm v$, we} present pseudo-code for the sum-product algorithm in Algorithm \ref{alg:FBA}. 
\blue{Note that, the normalization constant $Z = \sum_k \alpha_t \ind k \beta_t \ind k$ in Line \ref{line:normalize} of Algorithm \ref{alg:FBA} is time invariant \cite{forney2011partition}. }

\begin{algorithm}
\caption{Forward/Backward Algorithm}
\label{alg:FBA}
\begin{algorithmic}[1]
\Require{The transition matrix $P \in \R^{n \times n}$, the \red{observation message} $\bm \chi_t$ for $t \in \{1, \dots, T\}.$}
\Ensure{The inferred probability $\bm p_t$ for $t \in \{1, \dots, T\}$.}
\State Initialize the forward message $\bm \alpha_0 = \bm \pi$, the stationary distribution of $P$ 
\State \blue{Initialize the backward message $\beta_T\ind i = 1$, for all $i$}
\For {t from $2, \dots, T$} 
    \State $\bm \alpha_t = \diag(\bm \chi_t)P\bm \alpha_{t-1}$ 
\EndFor
\For {s from $T-1, \dots, 1$}
    \State \blue{$\bm \beta_s = P^{\intercal} \diag(\bm \chi_{s+1})\bm \beta_{s+1}$}
\EndFor
\For {t from $1, \dots, T$}
    \State $\bm p_t = \frac{\bm \alpha_t \odot \bm \beta_t}{\sum_k  \alpha_t \ind k \beta_t \ind k}$ \label{line:normalize}
\EndFor
\end{algorithmic}
\end{algorithm}

Fig.~\ref{fig:dis_to_cont} illustrates the sum-product algorithm's application to a time series of discretely sampled observations.
Consider a three-state-chain with symmetric transition rates as shown in the top-left of Figure \ref{fig:dis_to_cont}. 
If the system takes states 1 or 2, ``0'' will be observed, and ``1'' will be observed otherwise. 
Using the sum-product algorithm, we can find the probability of the system occupying each state given the discrete-time observations. 
As shown in the bottom row, as the sampling time step decreases, the conditional probabilities appear to converge to a smooth curve within each interval with a fixed observation (either ``0'' or ``1''). 
Intuitively, there should exist a continuous, perhaps piecewise differentiable, representation of the conditional probability of a partially observed process.  
We formalize this intuition below.


\begin{figure*}
    \centering
    \includegraphics[width=\textwidth]{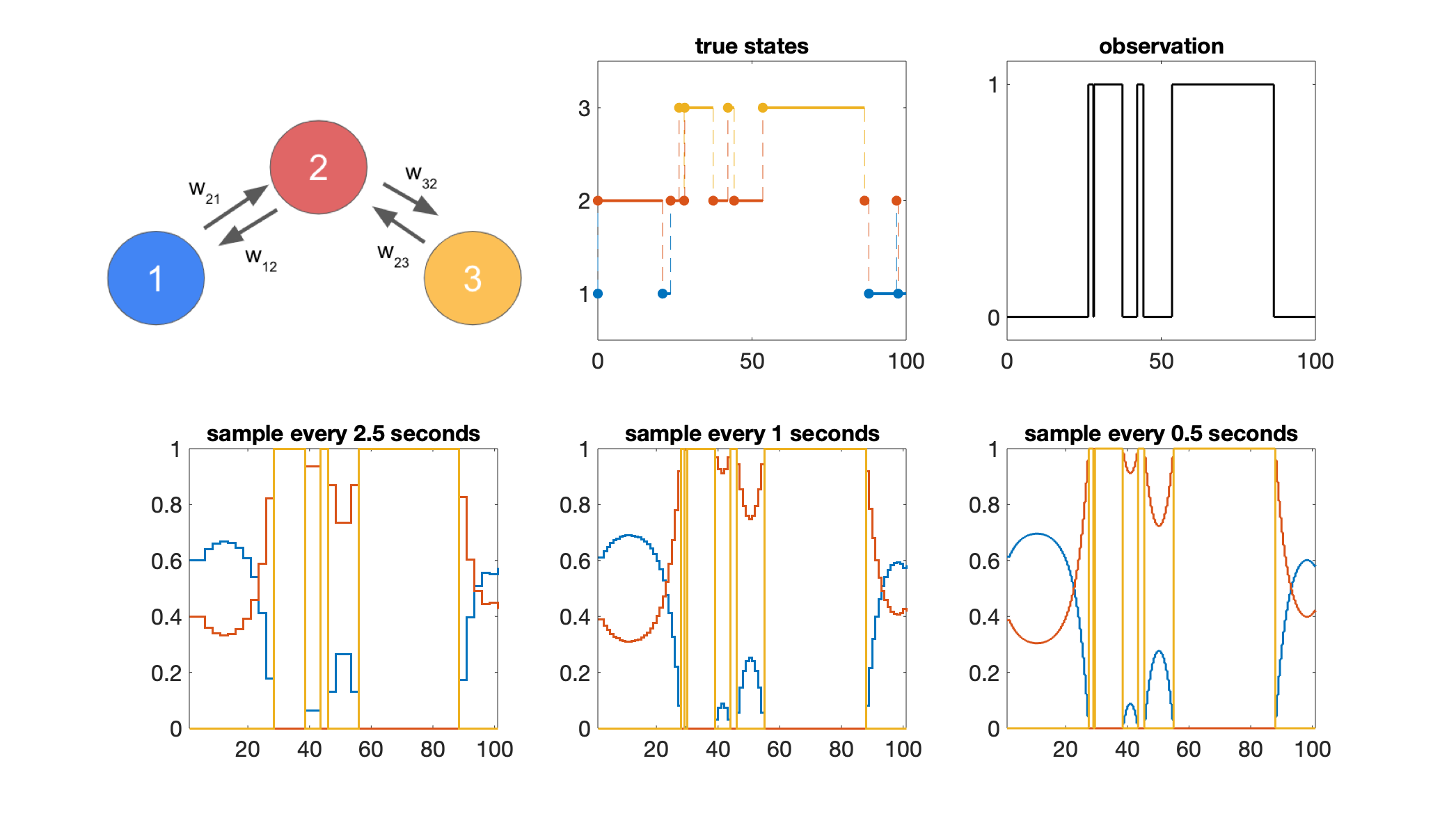}
    \caption{Illustrating convergence of conditional state occupancy probabilities to a differentiable function for a three-state model.
    (\textbf{Top left}) The state diagram. (\textbf{Top middle}) True simulated states. (\textbf{Top right}) Binary observation derived from true states. (\textbf{Bottom row}) Inference of hidden states via the sum-product algorithm with time steps 2.5 sec (\textbf{Bottom Left}), 1.0 sec (\textbf{Bottom Middle}) and  0.5 sec (\textbf{Bottom Right}).}
    \label{fig:dis_to_cont}
\end{figure*}

\section{Continuous-time Message Passing}\label{sec:main}

In this section, we present the main result, namely the derivation of the continuous-time message passing algorithm. In this new formalism, messages are passed in the form of linear differential equations on possible states given the observable system. In order to guarantee the existence of the continuous sum-product algorithm, we assume the following conditions.


\textbf{Assumptions:}
\begin{itemize}
    \item[A1] The continuous-time  process $\{S(t); ~t\in[0,T]\}$ takes values in a finite state space $\Omega=\{1,2,\dots,N\}$.
    \item[A2] $S(t)$ has the Markov property, and has exponentially distributed waiting times parameterized by a rate matrix $W$, with $w_{ji}$ specifying the transition rate from state $i$ to state $j$. \blue{Note that $W$ is constant in the interval $[0,T]$}.
    \item[A3] There is a distinguished subset $\mathcal{S}\subset \Omega$ such that the observable process $\mathcal{Y}(t)$ satisfies
    \begin{align*}
    \mathcal Y(t) = \begin{cases}
        1 & \IF S(t) \in \mathcal S, \\ 0 & \OW.
    \end{cases}
\end{align*}
\end{itemize}
\red{Further details appear in Appendix \ref{sec:prelim}.}

Under these assumptions, we obtain a continuous-time version of the sum-product algorithm by executing the following steps (made rigorous in the proof of Theorem \ref{thm:main} below).
Write out the matrix multiplication of the discrete-time algorithm element-wise. Focus on one sojourn where the observation doesn't change. Within that time interval, take the limit as the time step goes to zero to derive the continuous-time dynamics of the conditional probabilities. 
Extend the solution to the full time interval via appropriate boundary conditions at the transition between each sojourn.
The main result is stated in the theorem below.
    
\begin{theorem}\label{thm:main}
Suppose processes $S(t)$ and the associated process $\mathcal Y(t)$ satisfies assumptions A1, A2, and A3 above.
Then, given a realization of the process $\mathcal Y(t)$, the conditional probability $\bm p(t) = \Pr[S(t) | \mathcal Y(t)]$ exists, is piecewise smooth ($C^\infty$), and is $C^\infty$ on all intervals where $\mathcal Y(t)$ is constant. In particular, $\bm p(t)$ takes the form
\begin{align}
    \bm p(t) = \frac{\bm \rho(t)}{\sum_k \rho\ind k (t)} = \frac{\bm \alpha(t) \odot \bm \beta(t)}{\sum_k \alpha\ind k(t) \odot \beta\ind k(t)}
\end{align}
where $\odot$ denotes the \red{element-wise} product, \red{and} the quantities $\bm \alpha(t)$ and $\bm \beta(t)$ \red{are functions of time that} follow the linear ordinary differential equations
\begin{align}
    \fod {\alpha\ind i(t)} t &= \sum_{k \in \mathcal S \setminus \{i\}} w_{ki}\alpha\ind k(t) - \sum_{l \neq i} w_{il}\alpha\ind i(t) , \label{eq:forwarddiff}\\
    \fod {\beta\ind j(t)} t &= -\sum_{k \in \mathcal S \setminus \{j\}} w_{jk}\beta\ind k(t) + \sum_{l \neq j} w_{jl}\beta\ind j(t) . \label{eq:backwarddiff}
\end{align}
\end{theorem}

\begin{proof}
Without loss of generality, focus on the case where 
$\mathcal Y(0) = \mathcal Y(T) = 0$, and $\mathcal{Y}(t)=1$ for $0 < t < T$.
\red{We use $q(t)$ to denote a quantity, $q$, evolving in continuous time on the interval $[0,T]$, and $q_t$ to denote the same process sampled at discrete times.}

Let the time interval $[0,T]$ be discretized with a step size $\dt=T/n$ for some integer $n \gg 1$.
At each time step, $S(t)$ is sampled. Then, the sum-product algorithm can be used to solve for $\bm p_t$. Writing the matrix multiplication out yields the following set of equations for the forward and backward messages in discrete time:
\begin{align}
    \alpha_{t+\dt} \ind i &= \sum_k \Pr[s_{t+\dt} = i | s_t = k] \alpha_t\ind{k} \blue{ \chi_t\ind{i} }, \\
    \beta_t \ind i &= \sum_k \Pr[s_{t+\dt} = k | s_t = i] \beta_{t+\dt}\ind{k} \blue { \chi_{t+\dt} \ind{k} }.
\end{align} 
We neglect states not in $\mathcal S$ because only the probability conditioned on the observations is of interest. 
Then, for any state $i \in \mathcal S$, we argue that the corresponding forward message, $\bm \alpha\ind i(t)$, and backward message, $\bm \beta \ind i(t)$ in continuous time can be written as solutions of systems of differential equations, upon taking limits as $\dt\to 0$. 
For notational simplicity, for $t>\tau$ we let $\mathcal P_{t,\tau}\ind{i,j} = \Pr[S(t) = i | S(\tau) = j]$.
\begin{align}
    \nonumber\lefteqn{\fod {\alpha\ind{i}(t)} t }&\\
    &= \lim_{\dt\to 0} \frac{\alpha\ind{i}({t + \dt}) - \alpha\ind{i}(t)}{\dt} \\
    &= \lim_{\dt\to 0} \frac{1}{\dt} \Bigg[ \sum_k \mathcal P_{t+\dt,t}\ind{i,k} \alpha\ind{k}(t) \chi\ind{i}(t) - \alpha\ind{i}(t)\Bigg] \\
    \nonumber &= \lim_{\dt\to 0} \frac{1}{\dt} \biggr[ \sum_{k \in \mathcal S \setminus \{i\}} (w_{ki}\dt + o(\dt))\alpha\ind{k}(t) \ldots\\
    &\quad\quad\quad\quad\quad\quad - \sum_{l\neq i} (w_{il}\dt + o(t)) \alpha\ind i(t) \biggr] \\
    &= \sum_{k \in \mathcal S \setminus \{i\}} w_{ki}\alpha\ind k(t) - \sum_{l \neq i} w_{il}\alpha\ind i(t) 
\end{align}
\begin{align}
    \nonumber\lefteqn{\fod {\beta\ind{i}(t)} t}&\\ 
    &= \lim_{\dt\to 0} \frac{\beta \ind i (t + \dt) - \beta\ind{i}(t)}{\dt} \\
    \nonumber &= \lim_{\dt\to 0} \frac{1}{\dt} \Bigg[\beta\ind{i}({t+\dt})\ldots \\
    &\:\:\:\: - \sum_k \mathcal P_{t+\dt,t}\ind{k,i} \beta\ind{k}({t+\dt}) \blue{ \chi\ind{k}(t+\dt) } \Bigg] \\
    \nonumber &= \lim_{\dt\to 0} \frac{1}{\dt} \Bigg[ - \sum_{k \in \mathcal S \setminus \{i\}} (w_{ik}\dt + o(\dt))\beta\ind{k}({t+\dt})\\
    &\:\:\:\: + \sum_{l\neq i} (w_{il}\dt + o(t)) \beta\ind i({t+\dt}) \Bigg] \\
    &= -\sum_{k \in \mathcal S \setminus \{i\}} w_{ik}\beta\ind k(t) + \sum_{l \neq i} w_{il}\beta\ind i(t) 
\end{align}
\red{where $\lim_{\dt \to 0} o(\dt)/\dt = 0$. Readers could refer to Appendix \ref{sec:prelim} or consult existing literature such as \cite{gallager2013stochastic} for the relationship between transition probabilities and transition rates of a Markov jump process.} 

The conditional probability can be found by solving the differential equations, taking the component-wise product of $\bm \alpha (t)$ and $\bm \beta (t)$ for every $t \in [0,T]$, and normalizing so as to obtain a valid probability distribution. 

The differential equation formulation is only applicable \red{for the time intervals where the observation $\mathcal Y(t)$ is constant}. 
When the \red{observable $\mathcal Y(t)$} changes (when the systems transitions from a state in $\mathcal S$ to a state not in $\mathcal S$), the probability with respect to time might not be differentiable; in some cases, it is not even continuous. 
Therefore, we must specify boundary conditions to connect the probabilities from one sojourn to the next. 
The observable may change either by the system entering $S$ or else leaving $S$.
Suppose a transition occurred within time $(t_*-\dt, t_*]$ such that, for $t < t_*-\dt$, $S(t) \in \mathcal S$, and $S(t) \not \in \mathcal S$ for $t \geq t_*$. 
Call this event $E$. 
Then, we obtain the following transition rule for the forward message.
\begin{align}
    \Pr[&S(t_*) = j | E] \nonumber \\
    &= \dfrac{\sum_{i \in \mathcal S} \mathcal P_{t_*,t_*-\dt}^{j,i} \Pr[S(t_* - \dt) = i]}{\sum_{k \not \in \mathcal S} \sum_{i \in \mathcal S} \mathcal P_{t_*,t_*-\dt}^{k,i}\Pr[S(t_* - \dt) = i]} \\
    &= \red{ \dfrac{\sum_{i \in \mathcal S} (w_{ji}\dt + o(\dt))\Pr[S(t_* - \dt) = i]}{\sum_{k \not \in \mathcal S} \sum_{i \in \mathcal S} (w_{ki}\dt + o(\dt))\Pr[S(t_* - \dt) = i]} }\\
    &\to \dfrac{\sum_{i \in \mathcal S} w_{ji}\Pr[S(t_*^-) = i]}{\sum_{k \not \in \mathcal S} \sum_{i \in \mathcal S} w_{ki}\Pr[S({t_{*}^{-}}) = i]} \label{eq:forward-update}
\end{align}
as $\dt \to 0$. 
Here $\Pr[S(t_*^-) = i]$ is the probability of occupying state $i$ the instant before the transition, which can be found by solving the differential equations introduced above. 

We handle the boundary conditions at state transitions for the backward message similarly.
Define $E$ as above and let $s^*$ be a particular goal state. Then:
\begin{align}
    \Pr[& S(T) = s^* | S(t_*-\dt) = j, E] \nonumber \\
    &= \sum_{i \not \in \mathcal S} \Pr[S(T) = s^* , S(t_*) = i | S(t_* - \dt) = j, E] \nonumber \\
    &= \dfrac{\sum_{i \not \in \mathcal S} \Pr[S(T) = s^* | S({t_{*}}) = i] \mathcal P_{t_*, t_*-\dt}^{i,j}}{\sum_{k \in \mathcal S} \sum_{i \not \in \mathcal S} \mathcal P_{T,t_*}^{s^*,i}\mathcal P_{t_*, t_*-\dt}^{i,k}} \\
    &= \dfrac{\sum_{i \not \in \mathcal S} \Pr[S(T) = s^* | S({t_{*}}) = i] (w_{ij}\dt + o(\dt))}{\sum_{k \in \mathcal S} \sum_{i \not \in \mathcal S} \mathcal P_{T,t_*}^{s^*,i}(w_{ik}\dt + o(\dt))} \\ 
    &= \dfrac{\sum_{i \not \in \mathcal S} \Pr[S(T) = s^* | S({t_{*}}) = i](w_{ij} + \frac{o(\dt)}{\dt})}{\sum_{k \in \mathcal S} \sum_{i \not \in \mathcal S} \mathcal P_{T,t_*}^{s^*,i}(w_{ik} + \frac{o(\dt)}{\dt})}  \\
    &\to \dfrac{\sum_{i \not \in \mathcal S} w_{ij}\Pr[S(T) = s^* | S({t_{*}}) = i]}{\sum_{k \in \mathcal S} \sum_{i \not \in \mathcal S} w_{ik}\Pr[S(T) = s^*| S({t_*}) = i]} \label{eq:backward-update}
\end{align}
as $\dt \to 0$. 
Here, $\Pr[S(T) = s^*| S(t_*) = i]$ is a hitting probability associated with the backward message at time $t_*$, which can be found by solving the backward-message differential equation. 
These boundary conditions, together with the differential equations \eqref{eq:forwarddiff}-\eqref{eq:backwarddiff}, give the continuous-time evolution of the conditional probability for any finite-length observations.

Note that the equations \eqref{eq:forwarddiff}-\eqref{eq:backwarddiff} extend to the case where $\mathcal Y(0) = \mathcal Y(T) = 1$ and $\mathcal Y(t) = 0$ for $0 < t < T$  by viewing $\mathcal S \gets \Omega \setminus \mathcal S$. So, given a time series observation $\mathcal Y(t)$ where observation (the value of $\mathcal Y(t)$) changes at time $0 < t_1 < t_2 < \dots t_m$, we can solve for the analytic solution at any interval with consistent observation $(t_i, t_{i+1}]$ using equation \eqref{eq:forwarddiff} and \eqref{eq:backwarddiff}. Then, use the result to compute the initial condition for the next interval ---  namely, $(t_{i+1}, t_{i+2}]$ for the forward message and $(t_{i-1}, t_i]$ for the backward message --- as specified in \eqref{eq:forward-update} and \eqref{eq:backward-update}. \blue{Thus, the statement of Theorem \ref{thm:main} holds.}
\end{proof}

A general expression for the conditional probability can be obtained, but it is not of great utility in most systems. 
Yet, there are certain special cases that yield elegant solutions; we introduce several examples  in Section \ref{sec:analytic}.

The continuous-time sum-product algorithm follows directly from the derivation above, and is outlined in Algorithm \ref{alg:CFBA}. 
The discrete algorithm passes information through matrix multiplication of a truncated transition matrix; 
the continuous-time algorithm does the same by solving a system of linear differential equations using the truncated rate matrix. 
Since the final conditional probability is only piecewise differentiable, the boundary condition must be applied whenever a transition in or out of the observable set $\mathcal S$ occurs. 

\begin{algorithm}
\caption{Continuous-time Forward/Backward Algorithm}\label{alg:CFBA}
\begin{algorithmic}[1]
\Require{The rate transition matrix $W \in \R^{n\times n}$, the observed process $\mathcal Y(t)$.}
\Ensure{The inferred probability $\bm p(t)$ for $t \in [0, T]$.}
\State Let $[t_1, t_2, \dots, t_m]$ be a list of times where transitions occur
\State $\tau \gets 0$ 
\State $\bm \alpha^* \gets \bm \pi$, the stationary distribution
\For {j from $1, \dots, m$} 
    \State $\bm \alpha_j \gets $ solution to the forward message differential equation (Equation \ref{eq:forwarddiff}) from $\tau$ to $t_j$ with initial condition $\bm \alpha^*$ 
    \State $\tau \gets t_j$ 
    \State $\bm \alpha^* \gets$ distribution specified according to Equation \ref{eq:forward-update}
\EndFor
\State $\tau \gets T $
\State $\bm \beta^* \gets $ the uniform distribution 
\For {j from $m, \dots, 1$} 
    \State $\bm \beta_j \gets $ solution to the backward message differential equation (Equation \ref{eq:backwarddiff}) from $\tau$ to $t_j$ with initial condition $\beta^*$ 
    \State $\tau \gets t_j$
    \State $\bm \beta^* \gets$ distribution specified according to Equation \ref{eq:backward-update}
\EndFor
\State $\bm \alpha(t), \bm \beta(t) \gets$ concatenation $\bm \alpha_j$'s and $\bm \beta_j$'s
\State Compute $\bm \rho(t) = \bm \alpha(t) \odot \bm \beta(t)$, the component-wise product between $\bm \alpha(t)$ and $\bm \beta(t)$ pointwise with respect to $t$ 
\State Compute the conditional probability $\bm p(t) = \dfrac{\bm \rho(t)}{\sum_k \rho\ind k (t)}$
\end{algorithmic}

\end{algorithm}

\red{From a practical perspective, having the ability to solve for the conditional probabilities exactly through differential equations drastically lowers the computational cost of the forward/backward algorithm. 
Traditionally, the discrete-time algorithm propagates the forward and backward messages through matrix operations \emph{at each time step}. 
For long time-series and/or high-dimensional systems, this is computationally prohibitive. 
Through our continuous-time formalism, we solve the differential equations analytically, which is an operation that is independent of the length of the time-series, to find the forward or backward message at any time point. 
This difference effectively reduces the asymptotic scaling from $\mathcal O(\dt \inv)$ to $\mathcal O(1)$, with the later scaling only in the number of observable transitions. 
In scenarios where finding the appropriate boundary condition would require an iterative procedure of solving the forward and backward messages multiple times, our continuous time approach should be much more efficient than the traditional discrete-time method. 
We discuss the performance of the continuous-time message-passing algorithm further in Section \ref{sec:apps}}.

\section{Analytic Solution}\label{sec:analytic}

Theorem \ref{thm:main} in the previous section shows that the conditional probability is always available analytically upon normalizing the component-wise product of the forward and backward messages, \blue{ \textit{i.e.} $\bm p(t) = \bm \rho(t)/Z$ where $Z = \sum_k \rho\ind k(t)$. }
\blue{As in the discrete-time case,  the normalizing term $Z$ is time-invariant in the continuous-time case as well.  See Appendix \ref{sec:thmtwoproof}.
The conditional probability may therefore be expressed in a particularly elegant form in certain cases, namely as} the solution of a linear nonhomogeneous second-order differential equation.
\blue{We begin this section by considering two examples.
Following the examples, we consider extensions to higher dimensions.}

\subsection{Symmetric 3-State Chain}

\begin{figure}
    \centering
    \includegraphics[width=\linewidth]{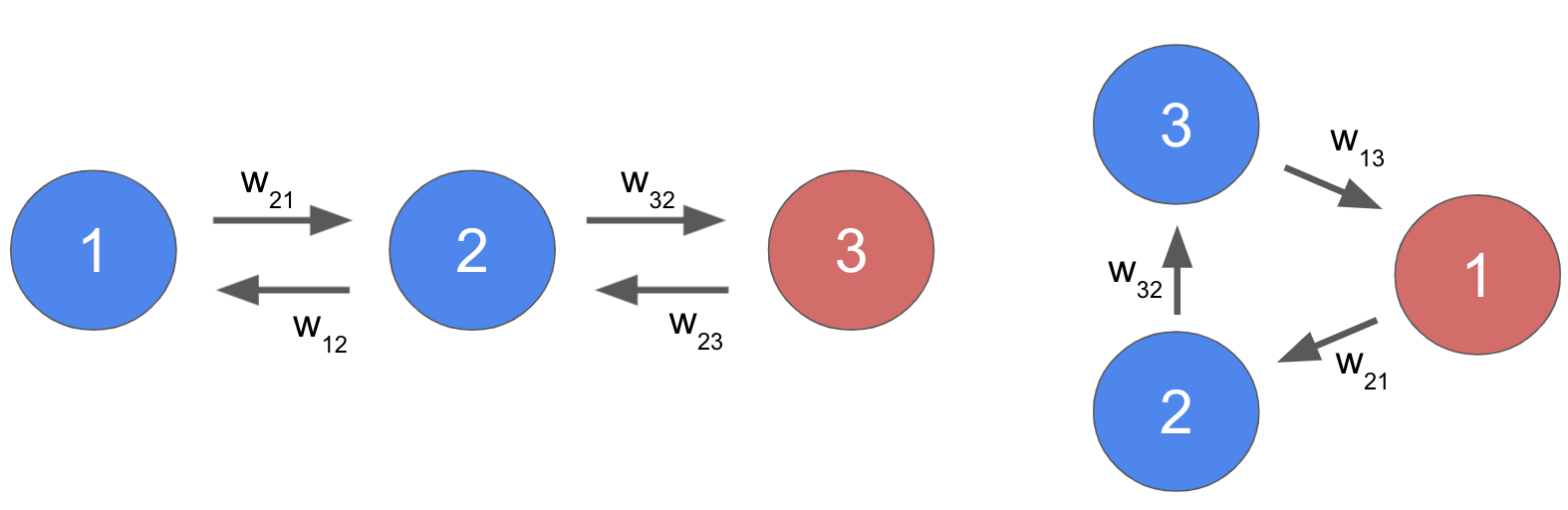}
    \caption{\blue{State diagrams of two systems for which the continuous-time message passing algorithm exhibit analytic simplifications}.
    \textbf{(Left)} 3-state chain with symmetric rates $w_{12}=w_{21}$.
    \textbf{(Right)} Irreversible 3-state loop. 
    States marked in red return $\mathcal{Y}(t)=1$ and blue return $\mathcal{Y}(t)=0$.
    }
    \label{fig:three-state}
\end{figure}

Consider the three-state chain depicted in the left panel of Figure \ref{fig:three-state}, where states 1 and 2 are hidden. 
Assume that the per-capita transition rates within the hidden block are symmetric, i.e.~$w_{12}=w_{21}>0$,
and assume $w_{13}=w_{31}=0$.
The rates $w_{23}>0$ and $w_{32}>0$ may be arbitrary.
These assumptions result in the following rate matrix:

\begin{align}
    W = \begin{pmatrix}
    -w_{21} & w_{12} & 0 \\ w_{21} & -(w_{12} + w_{32}) & w_{23} \\ 0 & w_{32} & -w_{23}
    \end{pmatrix}.
\end{align}

In this case, let $\mathcal S = \{3\}$, the singleton set of state 3. When $\mathcal Y(t) = 1$ the inference problem is trivial since the system takes state 3 with probability one. Thus, we emphasize the intervals when $\mathcal Y(t) = 0$. 

First, consider the forward message given by the following system of differential equations 
\begin{align} \label{eq:submatrix-chain}
    \fod {\bm \alpha}  t= \begin{pmatrix}
    -w_{21} & w_{12} \\ w_{21} & -(w_{12}+w_{32})
    \end{pmatrix} \bm \alpha~.
\end{align}
Note that the matrix defining the system of equations corresponds to the upper left block of $W$. To simplify notation, let $w_{21} = w_{12} = a$ and $w_{32} = b$. Then, the submatrix reduces to the following form:
\begin{align}
    \fod {\bm \alpha}  t= \begin{pmatrix}
    -a & a \\ a & -a-b
    \end{pmatrix} \bm \alpha~.
\end{align}

This is a linear system of differential equations that can be solved exactly. The sub-matrix is real-symmetric, so is diagonalizable. Therefore, the solution will be of the  form:
\begin{align}
    \bm \alpha(t) = A e^{\lambda_1 t} \bm v_1 + B e^{\lambda_2 t} \bm v_2.
\end{align}
where $\lambda_i$ is an eigenvalue of the rate submatrix and $\bm v_i$ is the corresponding eigenvector. The eigenvalues and vectors are:
\begin{align}
    \lambda_{1/2} &= a(-1 \pm \sqrt{1 + \gamma^2}) - \gamma, 
\end{align}
\begin{align}
    \bm v_{1/2} = \begin{pmatrix} \gamma \pm \sqrt{1 + \gamma^2} \\ 1 \end{pmatrix}, 
\end{align}
where $\gamma = b/2a$. Constants $A$ and $B$ are found through the initial condition given below. 
In the previous section, the initial forward message was set to the equilibrium distribution. However, the chain structure leaves no ambiguity in the state occupied when the observation changes from $1$ to $0$. So, the initial forward message is the delta distribution on state $2$:
\begin{align}
    \bm \alpha(0) = \begin{pmatrix}
    0 \\ 1
    \end{pmatrix}.
\end{align}

After some time $T$, the system re-enters the visible state, namely, state 3 again. By the same reasoning, we also have an unambiguous boundary condition for the backward message:
\begin{align}
    \bm \beta(T) = \begin{pmatrix}
    0 \\ 1
    \end{pmatrix}.
\end{align}
Also, by the symmetry in the rates, the backward message evolves according to the same equations as the forward message, but backward in time. Thus, $\bm \beta(t) = \bm \alpha(T-t)$,
\begin{align}
    \bm \beta(t) = Ae^{\lambda_1(T-t)} \bm v_1 + Be^{\lambda_2(T-t)} \bm v_2
\end{align}
where $\lambda_i$'s, $\bm v_i$'s, constants $A$ and $B$ all remain the same due to symmetry.

The conditional probability is proportional to the component-wise product $\bm \rho = \bm \alpha \odot \bm \beta$, which yields:
\begin{align}
    \bm \rho(t) &= \mathcal A \bigr[e^{\lambda_1(T-t)+\lambda_2t}+e^{\lambda_1 t + \lambda_2(T-t)}\bigr] + \mathcal B,
\end{align}
where
\begin{align}
    \mathcal A &= AB( \bm v_1\odot \bm v_2) = \red{\begin{pmatrix} -AB & AB \end{pmatrix}^\intercal},  \\
    \mathcal B &= A^2e^{\lambda_1 T} (\bm v_1\odot \bm v_1) + B^2e^{\lambda_2 T} (\bm v_2 \odot \bm v_2).
\end{align}
To recover the conditional probability, we must normalize such that the component sum evaluates to $1$. In this case, since the submatrix \red{in Equation \ref{eq:submatrix-chain}} is symmetric, the eigenvectors of the submatrix are orthogonal to each other, which means the sum of the components of $\bm v_1 \odot \bm v_2$, or the inner-product between $\bm v_1$ and $\bm v_2$, evaluates to 0. Then, the component-wise sum of $\mathcal{A}$ is zero, so the normalizing constant equals the component-wise sum of $\mathcal{B}$. 

Write $\bm p(t)$ to represent the conditional probability. 
We write $Z = \sum_k \mathcal B \ind k$ to represent the normalizing constant. 
The final equation describing the conditional probability can thus be rewritten in the following form:
\begin{align}
    \bm p(t) = \frac{1}{Z} \left(\mathcal A e^{(\lambda_1-\lambda_2)t}e^{\lambda_2T} + \mathcal A e^{(\lambda_2-\lambda_1)t}e^{\lambda_1T} + \mathcal B \right).
\end{align}
The components of $\bm p(t)$ may also be expressed as the solutions of a second-order ordinary differential equation, cf.~\eqref{eq:ode1}-\eqref{eq:sode}.

\subsection{Irreversible 3-State Loop}

Light-gated Channelrhodopsin-2 (ChR2) receptors can be modeled with a 3-state chain where each vertex has out-degree 1 and forms a cycle, as depicted in Figure \ref{fig:three-state}, right panel \cite{eckford2018channel, nagel2003channelrhodopsin}.  
Let state $1$ be open, and  states $2$ and $3$ be closed. That is, $\mathcal S = \{1\}$. 
The open/closed status of the channel is observed through voltage recordings: high conductance indicates that the channel is in the open state, and low conductance implies the closed state. The rate matrix can be written as the following:
\begin{align}
    W = \begin{pmatrix}
    -w_{21} & 0 & w_{13} \\ w_{21} & -w_{32} & 0 \\ 0 & w_{32} & -w_{13}
    \end{pmatrix}
\end{align}
where $w_{ji}$ is the transition rate from state $i$ to state $j$. Conditioning on the channel being closed, the message-passing algorithm takes the lower right $2\times 2$ block as the transition matrix.

When transitioning from $\mathcal S$ to a state not in $\mathcal S$, the system must enter state $2$ first and exit through state $3$. 
Thus the boundary conditions for the forward and backward message are
\begin{align}\label{eq:init-ex-b}
    \bm \alpha(0) = \begin{pmatrix}
    1 \\ 0
    \end{pmatrix}, ~~ \bm \beta(T) = \begin{pmatrix}
    0 \\ 1
    \end{pmatrix}
\end{align}
where the first component corresponds to state 2 and second state 3. Suppose $w_{32} \neq w_{13}$ and let $\gamma = \frac{w_{32}}{w_{13}-w_{32}}$.  Then, the solution of the message-passing differential equations with the initial condition enforced satisfies:
\begin{align}
    \bm \alpha(t) &= e^{-w_{32}t}\begin{pmatrix} 1 \\ \gamma \end{pmatrix} -\gamma e^{-w_{31}t}\begin{pmatrix} 0 \\ 1 \end{pmatrix}, \\
    \bm \beta(t) &= e^{-w_{32}(T-t)} \begin{pmatrix} 1 \\ 0 \end{pmatrix} + \frac{1}{\gamma}e^{-w_{13}(T-t)}\begin{pmatrix} -\gamma \\ 1 \end{pmatrix}.
\end{align}
To get the conditional probability, first take the component-wise product between the forward and backward message at time $t$.
\begin{align}
    \bm \rho(t) = e^{-w_{32}t - w_{13}(T-t)}\begin{pmatrix} -1 \\ 1 \end{pmatrix} + \begin{pmatrix} e^{-w_{32}T} \\ -e^{-w_{13}T} \end{pmatrix}
\end{align}
As in the previous example, the normalizing constant is invariant in time. 
In this case, $Z = e^{-w_{32}T} -e^{-w_{13}T}$. 
Thus, the time evolution of the conditional probability can be written  
\begin{align}
    \bm p(t) = \bm \rho(t)/Z.
\end{align}


\blue{Next, we consider a case in which the submatrix is not diagonalizable.}
\blue{We set $w = w_{32} = w_{13}$, so that} the corresponding submatrix:
\begin{align}
    U = \begin{pmatrix}
    -w & 0 \\ w & -w
    \end{pmatrix}
\end{align}
admits the following Jordan normal form:
\begin{align}
    U = \begin{pmatrix}
    0 & w \\ 1 & 0
    \end{pmatrix}
    \begin{pmatrix}
    -w & 1 \\ 0 & -w
    \end{pmatrix}
    \begin{pmatrix}
    0 & 1 \\ \frac{1}{w} & 0
    \end{pmatrix}.
\end{align}
Using the same initial condition as Equation \eqref{eq:init-ex-b} to solve the system of differential equations yields the following forward message.
\begin{align}
    \bm \alpha(t) = \begin{pmatrix}
    e^{-wt} \\ te^{-wt}
    \end{pmatrix}
\end{align}
We solve for the backward message by undergoing a similar procedure and yield
\begin{align}
    \bm \beta(t) = \begin{pmatrix}
    (T-t)e^{-w(T-t)} \\ e^{-w(T-t)}
    \end{pmatrix}.
\end{align}
Upon taking the component-wise product, we \blue{observe} that the evolution of conditional probability is independent of time:
\begin{align}
    \bm \rho(t) = e^{-wT} \begin{pmatrix}
    T - t \\ t
    \end{pmatrix}, ~~ Z = \frac{1}{Te^{-wT}}.
\end{align}
Since both the time entering and exiting the hidden states are fixed, we can see that the probability flows linearly from one state to the other at equal rates.

\blue{Time-invariant normalization is a fundamental property of the sum-product algorithm \cite{forney2011partition}.
In Appendix \ref{sec:thmtwoproof} we give an elementary demonstration of this property for continuous-time systems with time-homogeneous transition rates, as illustrated by the two cases presented above.}

\subsection{Generalization}

\blue{We show that under certain circumstances, the conditional probability follows a second-order nonhomogeneous linear ordinary differential equation.}



\begin{corollary}\label{cor:ODE}
For 3-state systems where the truncated submatrix $U$ has distinct eigenvalues, the conditional probability can be written as the solution of a nonhomogeneous second-order linear differential equation with constant coefficients.
\end{corollary}

\begin{proof}
Let $(\lambda_i, \bm v_i)$ be eigenpairs for $U$ and let $(\lambda_i, \bm w_i)$ be eigenpairs for $U^\intercal$. 
Then the forward and backward messages are given by the expressions:
\begin{align}
    \bm \alpha(t) &= Ae^{\lambda_1t} \bm v_1 + Be^{\lambda_2t} \bm v_2,  \\
    \bm \beta(t) &= Ce^{\lambda_1(T-t)} \bm w_1 + De^{\lambda_2(T-t)} \bm w_2.
\end{align}
Taking the component-wise product yields the form:
\begin{align}
    \bm \rho(t) &=  \mathcal A \bigr[ e^{\lambda_1 t + \lambda_2(T-t)} \bigr] + \mathcal B \bigr[e^{\lambda_1(T-t)+\lambda_2t} \bigr] + \mathcal C, \\
    &= \mathcal A \bigr[ e^{(\lambda_1 - \lambda_2)t}e^{\lambda_2T} \bigr] + \mathcal B \bigr[e^{(\lambda_2 - \lambda_1) t}e^{ \lambda_1T} \bigr] + \mathcal C,
\end{align}
where
\begin{align}
    \mathcal A &= AD( \bm v_1\odot \bm w_2), \quad \mathcal B = BC( \bm v_2\odot \bm w_1), \\
    \mathcal C &= ACe^{\lambda_1 T} (\bm v_1\odot \bm w_1) + BDe^{\lambda_2 T} (\bm v_2 \odot \bm w_2).
\end{align}
Since the left and right eigenvectors corresponding to different eigenvalues are orthogonal, all time-dependent terms cancel when normalizing. 
\blue{Thus the normalization constant} $Z = \sum_k \mathcal C\ind k$.
Upon differentiating $\bm p(t) = \bm \rho(t)/Z$ twice with respect to time, we  arrive at the second-order linear differential equation that describes the time evolution of the condition probability:
\begin{align}
\label{eq:ode1}
    \fodd {\bm p} t &= \frac{(\lambda_1-\lambda_2)^2}{Z} \bigr[ \mathcal A e^{(\lambda_1-\lambda_2)t}e^{\lambda_2T} + \mathcal B e^{(\lambda_2-\lambda_1)t}e^{\lambda_1T}\bigr] \\
    &= \frac{(\lambda_1-\lambda_2)^2}{Z} (\bm p - \mathcal C). \label{eq:sode}
\end{align}
This second order equation comes equipped with two boundary conditions set by fixing $\bm{\alpha}(0)$ and $\bm{\beta}(T)$.
\end{proof}

\begin{corollary}
For a system of arbitrary dimension, for which the truncated submatrix $U$ has exactly two distinct eigenvalues and is diagonalizable, the conditional probability can be written as a second-order linear differential equation of the same form as Equation \ref{eq:sode}.
\end{corollary}

\begin{proof}
Let $\lambda_1$ and $\lambda_2$ be the two distinct eigenvalues. Let $\bm v_i$ be the sum of all right eigenvectors corresponding to eigenvalue $\lambda_i$, and let $\bm w_i$ be the sum of all left eigenvectors corresponding to eigenvalue value $\lambda_i$. Then by diagonalizability, the proof for this corollary follows the proof of Corollary \ref{cor:ODE}.
\end{proof}

This result extends Corollary \ref{cor:ODE} to higher dimensional systems, under special conditions.
\red{These two corollaries match intuition: we are concerned with the state occupancy probability conditioned on both the entrance and exit times of an observable state, so the probability flow obeys a second-order differential equation, which also requires specifying two boundary conditions.}
We note, however, that  if the submatrix has \emph{three or more} distinct eigenvalues, we do not expect a similar result to hold. 
\red{Moreover, generic matrices typically have as many distinct eigenvalues as their dimension, so we do not expect this scenario to occur outside of special cases when the rate submatrix is highly regular. In addition, it would be difficult to check if a given rate submatrix satisfies these conditions through numerical solvers, as equality of repeated eigenvalues might be obscured by floating point arithmetic. Therefore, as much as the second-order differential equations interpretation is intuitively appealing, it is likely not a generic phenomenon that we should expect for arbitrary rate matrices.}

\section{Biological Example: CFTR} \label{sec:apps}

\begin{figure}
    \centering
    \includegraphics[width=.7\linewidth]{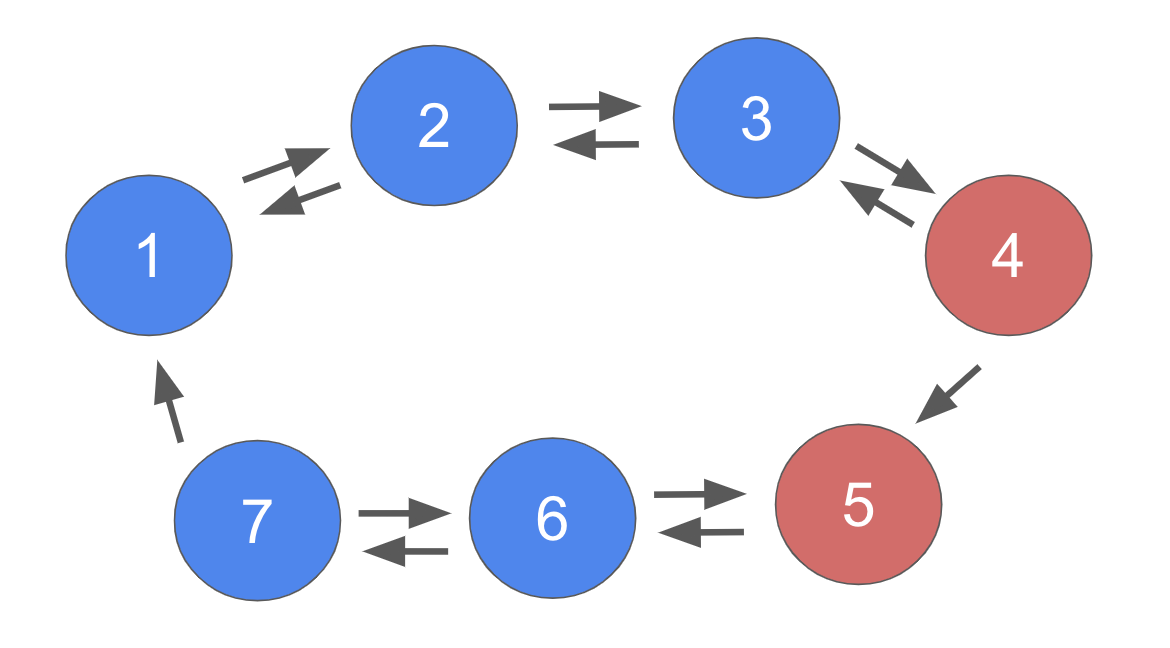}
    \caption{State diagram for CFTR. When the protein enters one of the red states, the ion channels will open and conduct a current. On the other hand, when CFTR is in one of the blue states, then the ion channels are closed and conducts no current. For this system $\mathcal{S}=\{4,5\}$.}
    \label{fig:cftr}
\end{figure}

\begin{figure*}
    \centering
    \includegraphics[width=\linewidth]{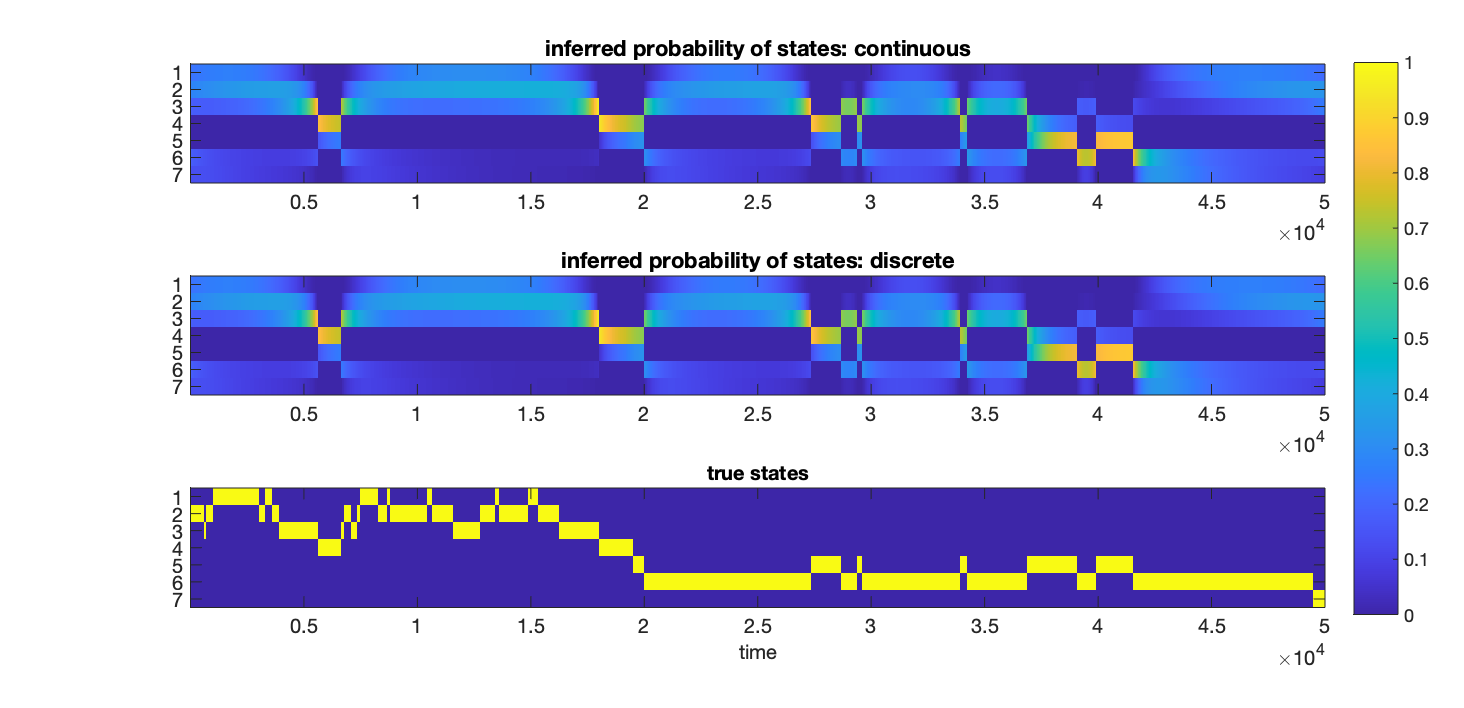}
    \caption{The inference result of the  continuous-time algorithm \textbf{(top)} and  the discrete-time algorithm  \textbf{(middle)} for  the true state occupancy simulated from the Gillespie algorithm \textbf{(bottom)}. Sampling time step  $\dt = 10^{-4}$.}
    \label{fig:inference}
\end{figure*}

Here, we show that our algorithm applies to a higher dimensional system, namely the 7-state model of the cystic fibrosis transmembrane conductance regulator (CFTR) protein. 
CF is a common life-threatening genetic disorder. 
CFTR is an important protein that regulates the opening and closing of ion channels. 
Loss of CFTR function causes pancreatic insufficiency as well as airway infection due to excessive mucus, 
which in turn can cause a variety of complications such as impaired innate immunity and respiratory failure \cite{goetz2019review}. 

Mathematically, the behavior of CFTR can be captured with a 7-state hidden Markov process.
Figure \ref{fig:cftr} illustrates the state diagram \cite{fuller2005block}. 
The ion channel opens (conducts an ionic current) when CFTR is in state four or five (marked in red).
When in the nonconducting states (marked in blue) the ion channel is closed.
Transmembrane conductance recordings report whether CFTR is in a conducting or a nonconducting state, but do not directly indicate which of the possible states is occupied. 
Using the numbering shown in Figure \ref{fig:cftr}, the rate matrix is defined as follows:

\begin{footnotesize}
\begin{align}
    W = \begin{pmatrix*}[r] -9.9 & 5.0 & 0 & 0 & 0 & 0 & 1.7 \\
    9.9 & -12.7 & 5.8 & 0 & 0 & 0 & 0 \\
    0 & 7.7 & -10.7 & 10.0 & 0 & 0 & 0 \\
    0 & 0 & 4.9 & -17.1 & 0 & 0 & 0 \\
    0 & 0 & 0 & 7.1 & -3.0 & 7.0 & 0 \\ 
    0 & 0 & 0 & 0 & 3.0 & -13.0 & 12.8 \\
    0 & 0 & 0 & 0 & 0 & 6.0 & -14.5 \end{pmatrix*}.
\end{align}
\end{footnotesize}

\normalsize

We used Gillespie's exact stochastic simulation algorithm  \cite{gillespie1977exact,gillespie2007stochastic,wilkinson2018stochastic} to generate sample traces of the ion channel states and recordings. 
The simulation produces discrete state, continuous time trajectories.
We introduced a finite sampling time step to discretize the simulated data along the time axis, consistent with data obtained through experimental recordings.  
Then, we place the generated data into discretized time bins where the size of each bin (one can think of this as the sampling time step) is a parameter that can be altered. 
We solved the differential equations giving the forward and backward messages exactly via function handles in Matlab.
Figure \ref{fig:inference} compares the traces produced by the classical discrete-time algorithm, our continuous-time algorithm, and the true states, and shows excellent agreement among the respective curves.

With a sufficiently small time step ($\dt=10^{-4}$), the continuous and discrete-time algorithms show no visible discrepancy, as expected. 
Figure \ref{fig:convergence} shows that the maximum discrepancy ($\ell_\infty$ norm) between the conditional probabilities generated by the continuous time and discrete time algorithms decreases linearly with the sampling time step $\dt$.
The figure shows the results from an ensemble of \red{forty} independent repeated trials for each sample size.
\red{Again, we emphasize} that while the discrete time algorithm requires iterating over all samples, which scales $\mathcal O(\dt\inv)$, the continuous time algorithm outputs a \textit{description} of the conditional probability --- that is, a function that returns a value at an exact time point queried --- in $\mathcal O(1)$ time, scaling only with the number of transitions. 

\begin{figure}
    \centering
    \includegraphics[width=\linewidth]{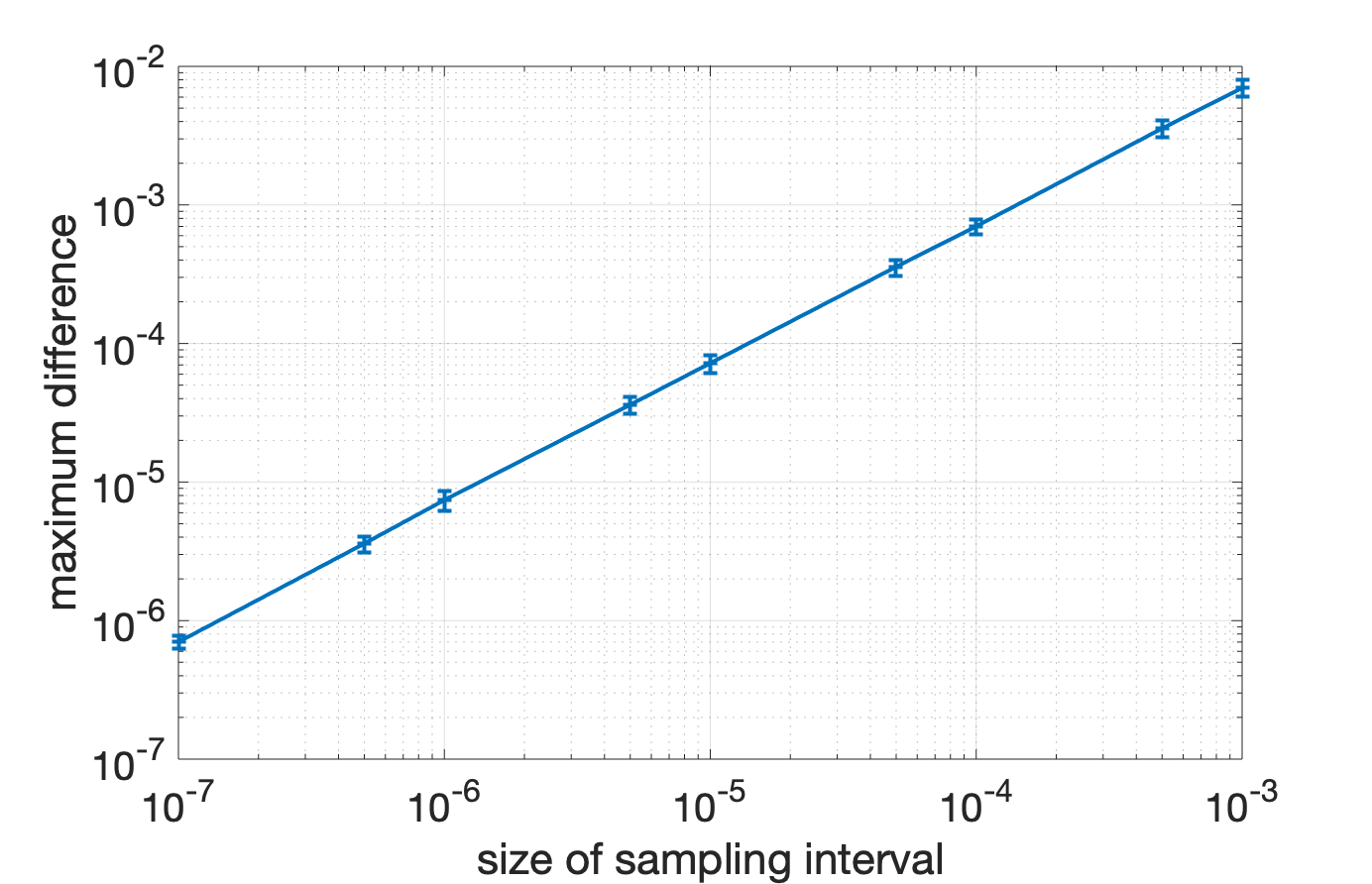}
    \caption{First-order convergence of the message passing algorithms as a function of sampling time step. 
    The mean maximum difference was plotted along with one standard deviation about the mean as the error bar for an ensemble of \red{forty} trajectories.}
    \label{fig:convergence}
\end{figure}

\section{Discussion \& Conclusions}\label{sec:conclusion}
This paper presents an algorithm for continuous-time inference on partially observable Markov processes with discrete state spaces. 
We show that the well-known sum-product algorithm can be extended to the continuous-time domain via two sets of differential equations and pointwise normalization. 
\blue{In the continuous time setting, we were able to solve the trajectory of conditional probabilities exactly given a finite time-series.
Moveover, we find that the dynamics of the state occupancy probabilities can be reduced to second-order differential equations under special circumstances.}
These results are valuable not only for their mathematical interest,
but also because they have the potential to reduce the inference problem to solving systems of linear differential equations, 
with a potentially significant reduction in computational complexity for long time series.
%
Numerically, the continuous-time algorithm is consistent with the discrete-time algorithm in the limit of small time step, but executes in approximately constant time rather than linearly in the number of time steps.  

Our formalism extends naturally to non-binary observations. 
Briefly, let $\mathcal S_1, \mathcal S_2, \dots, \mathcal S_m$ be a partition over the sample space $\Omega$, and suppose the observable process $\mathcal Y$ is given as:
\begin{align}
    \mathcal Y(t) = i ~~\IF S(t) \in \mathcal S_i.
\end{align}
For a sojourn with observation $i$, we can apply the forward and backward message-passing scheme as introduced for the binary case, viewing the observation as either in $\mathcal S_i$ or $\Omega \setminus \mathcal S_i$. 
At the boundaries, the same technique (as in the proof of Theorem \ref{thm:main}) can be used for finding the update rule by noting the possible transitions from $\mathcal S_i$ to $\mathcal S_j$, for all $j \neq i$. 
This approach can finally be extended to an arbitrary collection of subsets of $\Omega$ where the elements are not necessarily disjoint. 
One may accomplish this extension by expanding all unions and intersections as disjoint sets possibly with the same observable. 
We defer detailed investigations in this direction to future work.

\blue{Extending the continuous-time message passing to the case with inhomogeneous transition rates remains an open problem. 
We expect a derivation similar to the proof of the main theorem would be applicable, possibly with smoothness constraints on the transition rates. 
While the normalization constant will remain invariant in time, we do not expect a result such as Corollary \ref{cor:ODE} to hold beyond constant transition rates.}

Finally, we note the relationship between conditional probability and a second-order differential equation is intuitively satisfying: the entry and exit times act as two boundary conditions that fix the endpoints of the evolution, whereas the unconstrained forward evolution equation, a first-order differential equation, requires only the starting condition. 
It is an interesting question for future work to investigate under which assumptions  a similar result as Corollary \ref{cor:ODE} would hold for systems with more than two distinct eigenvalues.

\appendices
\section{Notation and Preliminaries}\label{sec:prelim}

For completeness, we define a continuous-time discrete-space Markov process with exponential waiting times.
\begin{definition}[Markov Process] Let $\{S(t); t \in [0,T]\}$ be a discrete-space, continuous-time stochastic process where for each $t$, $S(t)$ is a random variable with state space $\Omega = \{1,2,\dots, N\}$, $N < \infty$. Then, the process $S(t)$ has the Markov Property if for any $s_1, s_2, \dots, s_m \in \Omega$, $0 < \tau_1 < \tau_2 < \dots < \tau_{m-1} < t$,
\begin{align}
        \nonumber
        \Pr[S(t) = s_m | S(\tau_1) = s_1, S(\tau_2) = s_2, \dots, S(\tau_{m-1}) = s_{m-1}] \\
        = \Pr[S(t) = s_m | S(\tau_{m-1}) = s_{m-1}].
\end{align}\normalsize
Furthermore, the process has exponential waiting times if for any $i,j \in \Omega$, $i \neq j$, there is a constant rate $0 \leq w_{ji} < \infty$ such that 
\begin{align*}\footnotesize
    \begin{cases}
        &\Pr[S(t + \dt) = i | S(t) = i] = 1 - \sum_{k \neq i} w_{ki}\dt + o(\dt) \\
        &\Pr[S(t + \dt) = j | S(t) = i] = w_{ji}\dt + o(\dt) \\
        &\Pr[S(t+\dt)=j , S(t + 2\dt) = k | S(t) = i] = o(\dt), ~~ \forall k
    \end{cases}
\end{align*}\normalsize
for sufficiently small $\dt$.
\end{definition}

The following table lists notation used in the paper. 
\begin{figure}[h]
    \centering
    \begin{tabular}{c|c}
        Symbol & Meaning \\
        \hline
        $S_t$ & discrete-time process with (integer) time index $t$ \\
        $S(t)$ & continuous-time process with time index $t$ \\
        $\Omega$ & sample space of a process at fixed time \\
        $\mathcal S$ & states that give observable ``1'' \\
        $\mathcal Y(t)$ & the observed process (indicating $S(t) \in \mathcal S$) \\
        $P$ & transition matrix of a Markov chain \\
        $W$ & transition rate matrix of Markov process \\
        $\bm \alpha_t; \bm \alpha(t)$ & forward message \\
        $\bm \beta_t; \bm \beta(t)$ & backward message \\
        $\bm \chi_t; \bm \chi(t)$ & \red{observation message} \\
        $\bm \rho_t; \bm \rho(t)$ & unnormalized conditional state-occupancy probability  \\
        $Z(t)$ & normalizing constant for $\bm \rho(t)$ \\
        $\bm p_t; \bm p(t)$ & conditional state-occupancy probability\\
        $\dt$ & sampling time step
    \end{tabular}
\end{figure}

\section{Invariant normalization in continuous time} \label{sec:thmtwoproof}

\blue{Time-invariant invariant normalization is a fundamental property of the sum-product algorithm \cite{forney2011partition}.
In the continuous time case, we observed that this property can be easily confirmed when the submatrices are diagonalizable, because the left and right eigenvectors corresponding to different eigenvalues are orthogonal. 
When the submatrix has nontrivial Jordan blocks, the time-dependent term cancels in less obvious ways.
Here we show that time-independent normalization holds in general, using only elementary calculus without  invoking the machinery of factor-graphs.}


Recall that the normalizing constant is $Z(t) = \sum_i \alpha \ind i (t) \beta \ind i (t)$. Using Equation \ref{eq:forwarddiff} and \ref{eq:backwarddiff}, we can obtain the following series of expressions. 
\begin{align}
    \fod {Z(t)} t &= \od t \sum_{i \in \mathcal S} \alpha \ind i (t) \beta \ind i (t) \\
    &= \sum_{i \in \mathcal S} \fod {\alpha \ind i (t)} t \beta \ind i (t) + \alpha \ind i (t) \fod {\beta \ind i (t)} t \\
    &=  \sum_{i \in \mathcal S} \left( \sum_{k \in \mathcal S \setminus \{i\}} w_{ki} \alpha \ind k  - \sum_{l \neq i} \alpha \ind i \right) \beta \ind i \ldots \nonumber \\
    &\quad \quad + \alpha \ind i \left( - \sum_{k \in \mathcal S \setminus \{i\}} w_{ik} \beta \ind k  + \sum_{l \neq i} \beta \ind i \right) \\
    &= \sum_{i \in \mathcal S} \sum_{k \in \mathcal S \setminus \{i\}} w_{ki} \alpha \ind k \beta \ind i - w_{ik} \alpha \ind i \beta \ind k \\
    &= 0
\end{align}
Since the rate of change of $Z(t)$ is zero, the normalizing constant is invariant of time.


\bibliographystyle{IEEEtran}
\bibliography{IEEEabrv,ref}

\end{document}